\documentclass[onecolumn,draftclsnofoot, 12pt, ]{IEEEtran}
\usepackage{times}
\usepackage[T1]{fontenc}
\usepackage{graphicx}
\usepackage{amsmath}
\usepackage{setspace}
\usepackage{amsbsy}
\usepackage{amsthm}
\usepackage{cite}
\usepackage{cite}
\usepackage{authblk}
\usepackage{subfigure}
\usepackage{amsfonts}
\usepackage{array}
\usepackage{algorithm}
\usepackage{algorithmicx}
\usepackage{algpseudocode}
\usepackage{breqn}
\newtheorem{theorem}{Theorem}
\newtheorem{lemma}[theorem]{Lemma}

\begin{document}
\title{ Linear Precoding Gain for Large MIMO Configurations
 with QAM and Reduced Complexity}
\author{Thomas~Ketseoglou,~\IEEEmembership{Senior~Member,~IEEE,} and Ender~Ayanoglu,~\IEEEmembership{Fellow,~IEEE}

\thanks{T. Ketseoglou is with the Electrical and Computer Engineering Department, California State Polytechnic University, Pomona, California (e-mail: tketseoglou@csupomona.edu). E. Ayanoglu is
with the Center for Pervasive Communications and
Computing, Department of Electrical Engineering and Computer Science,
University of California, Irvine (e-mail: ayanoglu@uci.edu).}
}


\maketitle
\begin{abstract}
In this paper, the problem of designing a linear precoder for Multiple-Input Multiple-Output (MIMO) systems in conjunction with Quadrature Amplitude Modulation (QAM) is addressed. First, a novel and efficient methodology
to evaluate the input-output mutual information for a general Multiple-Input Multiple-Output (MIMO) system as well as its corresponding gradients is presented, based on the Gauss-Hermite quadrature rule. Then, the method is exploited in a block coordinate gradient ascent optimization process to determine the globally optimal linear precoder with respect to the MIMO input-output mutual information for QAM systems with relatively moderate MIMO channel sizes. The proposed methodology is next applied in conjunction with the complexity-reducing per-group processing (PGP) technique to both perfect channel state information at the transmitter (CSIT) as well as statistical channel state information (SCSI) scenarios, with large transmitting and receiving antenna size, and for constellation size up to $M=64$.
We show by numerical results that the precoders developed offer significantly better performance than the configuration with no precoder as well as the maximum diversity precoder for QAM with constellation sizes $M=16,~32$, and $~64$ and for MIMO channel size up to $100\times100$.
\end{abstract}
\IEEEpeerreviewmaketitle
\section{{INTRODUCTION}}
The concept of Multiple-Input Multiple-Output (MIMO) systems still represents a prevailing research direction in wireless communications due to its ever increasing capability to offer higher rate, more efficient communications, as measured by spectral utilization, and under low transmitting or receiving power. Within MIMO research,
the problem of designing an optimal linear precoder toward maximizing the mutual information between the input and output was considered in \cite{Verdu1,Verdu2} where the optimal power allocation strategies are presented (e.g., Mercury Waterfilling (MWF)), together with general equations for the optimal precoder design. In addition, \cite{Payaro} also considered precoders for mutual information maximization and showed that the left eigenvectors of the optimal precoder can be set equal to the right eigenvectors of the channel. Finally, in \cite{Fettweis}, a mutual information maximizing precoder for a parallel layer MIMO detection system is presented reducing the performance gap between maximum likelihood and parallel layer detection.

Recently, globally optimal linear precoding techniques were presented \cite{Xiao,Lamarca} for scenarios employing perfect channel state information available at the transmitter (CSIT)\footnote{Under CSIT the transmitter has perfect knowledge of the MIMO channel realization at each transmission.} with finite alphabet inputs, capable of achieving mutual information rates
much higher than the previously presented MWF \cite{Verdu1} techniques by introducing input symbol correlation through a unitary input
transformation matrix in conjunction with channel weight adjustment (power allocation). In addition, more recently, \cite{Max} has presented an iterative algorithm for precoder optimization for sum rate maximization of  Multiple Access Channels (MAC) with Kronecker MIMO channels. Furthermore, more recent work has shown that when only statistical channel state information (SCSI)\footnote{SCSI pertains to the case in which the transmitter has knowledge of only the MIMO channel correlation matrices \cite{Khan,Weich} and the thermal noise variance.} is available at the transmitter, in asymptotic conditions when the number of transmitting and receiving antennas grows large, but with a constant transmitting to receiving antenna number ratio, one can design the optimal precoder by looking at an equivalent constant channel and its corresponding adjustments as per the pertinent theory \cite{SCSI}, and applying a modified expression for the corresponding ergodic mutual information evaluation over all channel realizations. This development allows for a precoder optimization under SCSI in a much easier way \cite{SCSI}. However, existing research in the area does not provide any results of optimal linear precoders in the case of QAM with constellation size $M\geq 16$, with the exception of \cite{Lozano}. In past research work, a major impediment toward developing optimal precoders for QAM has been a lack of an accurate and efficient technique toward input-output mutual information evaluation, its gradients, and evaluation of the input-output minimum mean square error (MMSE) covariance matrix, as required by the precoder optimization algorithm and other algorithms involved, e.g., the equivalent channel determination in the SCSI case \cite{SCSI}.

In this paper, we propose optimal linear precoding techniques for MIMO, suitable for QAM with constellation size $M\geq 16$. An additional advantage of these techniques is their ability to accommodate MIMO configurations with very large antenna sizes, e.g., $100\times 100$. The only related work in this area is \cite{Lozano} which has antenna sizes up to $32\times 32$. We show in the sequel that the proposed method is faster than the one in \cite{Lozano}. Carrying out this calculation  has been very difficult to do until now due to the complexity involved in tackling this problem. Our approach entails a novel application of the Hermite-Gauss quadrature rule \cite{G_H} which offers a very accurate and efficient way to evaluate the capacity of a MIMO system with QAM. We then apply this technique within the context of a block gradient ascent method \cite{Boyd} in order to determine the globally optimal linear precoder for MIMO systems, in a similar fashion to \cite{Xiao}, for systems with CSIT and small antenna size. We show that for $M=16,~32$, and $64$ QAM, the optimal linear precoder offers $50\%$ better mutual information than the maximal diversity precoder (MDP) of \cite{Giann_max_div} and the no-precoder case, at low signal-to-noise ratio (SNR) for a standard $2\times2$ MIMO channel, however the absolute utilization gain achieved is lower than $1~b/s/Hz$. We then proceed to show that significantly higher gains are available for different channels, e.g., a utilization gain of $1.30~b/s/Hz$ at $SNR=10~dB$, when $M=16$. We then employ larger antenna configurations, e.g., up to $40\times40$ with CSIT and $M=16,~32,$ and $64$ together with the complexity reducing technique of per-group processing (PGP) which was originally presented in \cite{TE_TWC}, and show very high gains available with reduced system complexity. Finally, we also employ SCSI scenarios in conjunction with PGP and show very significant gains for large antenna sizes, e.g., $100 \times 100$ and $M=16,~32, ~64$. The main advantages of our work compared with other interesting proposals for large MIMO sizes, e.g., \cite{Lozano}, lie over three main directions: a) It offers a globally optimal precoder solution for each subgroup, instead of a locally optimal one, b) It is faster, c) It allows for larger constellation size, e.g., $M=32,~64$, and d) It allows larger MIMO configurations, e.g., $100\times 100$.

The paper is organized as follows: Section II presents the system model and problem statement. Then, in Section III, we present a novel Gauss-Hermite approximation to the evaluation of the input-output mutual information of a MIMO system that allows for fast, but otherwise very accurate evaluation of the input-output mutual information of a MIMO system, and thus represents a major facilitator toward determining the globally optimal linear precoder for LDPC MIMO. In Section IV, we present numerical results for the globally optimal precoder that implements the Gauss-Hermite approximation in the block coordinate gradient ascent method. Finally, our conclusions are presented in Section V.

\section{{SYSTEM MODEL AND PROBLEM STATEMENT}}
The instantaneous $N_t$ transmit antenna, $N_r$ receive antenna MIMO model is described by the following equation
\begin{eqnarray}
{\mathbf y} = {\mathbf H}{\mathbf G}{\mathbf x}+{\mathbf n}, \label{eq_1}
\end{eqnarray}
where ${\mathbf y}$ is the $N_r\times 1$ received vector, ${\mathbf H}$ is the $N_r\times N_t$ MIMO channel matrix,
 ${\mathbf G}$ is the precoder matrix of size $N_t\times N_t$, ${\mathbf x}$ is the $N_t\times 1$ data vector with independent components
each of which is in the QAM
constellation of size $M$, ${\mathbf n}$ represents the circularly symmetric complex Additive White Gaussian Noise (AWGN) of size $N_r\times 1$, with mean zero and
covariance matrix ${\mathbf K}_n = \sigma^2_n {\mathbf I}_{N_r}$, where ${\mathbf I}_{N_r}$ is the $N_r\times N_r$ identity matrix, and $\sigma^2 =\frac{1}{SNR}$, $SNR$ being the
(coded) symbol signal-to-noise ratio. In this paper, a number of different channels will be considered, e.g., channels comprising independent complex Gaussian components or spatially correlated Kronecker-type channels \cite{Khan} (including those similar to the 3GPP spatial correlation model (SCM) \cite{SCM}), or more generally Weichselberger channels \cite{Weich}. The precoding matrix ${\mathbf G}$ needs to satisfy the following power constraint
\begin{equation}
\mathrm {tr}({\mathbf G}{\mathbf G}^h)=N_t,
\end{equation}
where $\mathrm {tr}({\mathbf A})$,  ${\mathbf A}^h$ denote the trace and the Hermitian transpose of matrix ${\mathbf A}$, respectively.
An equivalent model called herein the ``virtual'' channel is given by \cite{Xiao}
\begin{equation}
{\mathbf y} = {\boldsymbol \Sigma}_H {\boldsymbol \Sigma}_G {\bf V}_G^h {\mathbf x} + {\mathbf n}, \label{eq_eq}
\end{equation}
where ${\boldsymbol \Sigma}_H$ and ${\boldsymbol \Sigma}_G$ are diagonal matrices containing the singular values of ${\mathbf H}$, ${\mathbf G}$, respectively and
${\mathbf V}_G$ is the matrix of the right singular vectors of ${\mathbf G}$. When LDPC is employed in this MIMO system, the overall utilization in $b/s/Hz$ is determined by the mutual information between the transmitting branches ${\mathbf x}$ and the receiving ones, ${\mathbf y}$ \cite{tenBrink,Alexei}. It is shown \cite{Xiao} that the mutual information between ${\mathbf x}$ and ${\mathbf y}$, for channel realization ${\mathbf H}$,
$I({\mathbf x};{\mathbf y})$,
is only a function of ${\mathbf W} = {\mathbf V}_G {\boldsymbol \Sigma}_H^2 {\boldsymbol \Sigma}_G^2 {\mathbf V}_G^h$.
The optimal CSIT precoder ${\mathbf G}$ is found by solving:
\begin{eqnarray}
\begin{aligned}
& \underset{\mathbf G}{\text{maximize}}
& & I({\mathbf x};{\mathbf y})\\
& \text{subject to}
& &  \mathrm {tr}({\mathbf G} {\mathbf G}^h) = N_t, \\ \label{eq_orig}
\end{aligned}
\end{eqnarray}
called the ``original problem,'' and
\begin{eqnarray}
\begin{aligned}
& \underset{{\mathbf V}_G, {\boldsymbol \Sigma}_G}{\text{maximize}}
& & I({\mathbf x};{\mathbf y})\\
& \text{subject to}
& &  \mathrm {tr}({\boldsymbol \Sigma}_G^2 ) = N_t, \label{eq_global}
\end{aligned}
\end{eqnarray}
called the ``equivalent problem,'' where the reception model of (\ref{eq_eq}) is employed.
The solution to (\ref{eq_orig}) or (\ref{eq_global}) results in exponential complexity at both transmitter and receiver, and it becomes especially difficult for QAM with constellation size $M \geq 16$ or large MIMO configurations. A major difficulty in the QAM case stems from the fact that there are multiple evaluations of $I({\mathbf x};{\mathbf y})$ in the block coordinate ascent method employed for determining the globally optimal precoder. More specifically, for each block coordinate gradient ascent iteration, there are two line backtracking searches required \cite{Xiao}, which demand one $I({\mathbf x};{\mathbf y})$ plus its gradient evaluations per search trial, and one additional evaluation at the end of a successful search per backtracking line search. Thus, the need of a fast, but otherwise very accurate method of calculating $I({\mathbf x};{\mathbf y})$ and its gradients prevails as instrumental toward determining the globally optimal linear precoder for CSIT. In the SCSI case, the corresponding optimization problem becomes
\begin{eqnarray}
\begin{aligned}
& \underset{\mathbf G}{\text{maximize}}
& & {\mathbb E}_{\mathbf H}\left\{I({\mathbf x};{\mathbf y})\right\}\\
& \text{subject to}
& &  \mathrm {tr}({\mathbf G} {\mathbf G}^h) = N_t, \\ \label{eq_SCSI}
\end{aligned}
\end{eqnarray}
where the expectation is performed over all the channels ${\mathbf H}$. The ground-breaking work of \cite{SCSI} has shown that the problem in (\ref{eq_SCSI}) for large antenna sizes can be solved by an approximate way of calculating the ergodic mutual information ${\mathbb E}_{\mathbf H}\left\{I({\mathbf x};{\mathbf y})\right\}$ for a fixed precoding matrix ${\mathbf G}$ through well-determined parameters of a deterministic channel. These parameters include the mutual information of the corresponding deterministic channel, i.e., a CSIT scenario. Thus, methods that offer simplification of CSIT mutual information evaluation,  $I({\mathbf x};{\mathbf y})$, are also important in the SCSI case toward determining the globally optimal linear precoder for LDPC MIMO.

\section{ACCURATE APPROXIMATION TO $I({\mathbf x};{\mathbf y})$ FOR MIMO SYSTEMS BASED ON GAUSS-HERMITE QUADRATURE}
In Appendix A we prove that by applying the Gauss-Hermite quadrature theory for approximating the integral of a Gaussian function multiplied with an arbitrary real function $f(x)$, i.e.,
\begin{equation}
F \doteq \int_{-\infty}^{+\infty}\exp(-x^2)f(x)dx,
\end{equation}
which is approximated in the Gauss-Hermite approximation with $L$ weights and nodes as
\begin{equation}
F \approx \sum_{l=1}^L c(l) f(v_l) = {\mathbf c}^t {\mathbf f},
\end{equation}
with ${\mathbf c}=[c(1) \cdots c{(L)}]^t$, $\{v_l\}_{l=1}^L$, and ${\mathbf f} = [f(v_1) \cdots f(v_L)]^t,$ being the vector of the weights, the nodes, and function node values, respectively (see Appendix A),
a very accurate approximation is derived for $I({\mathbf x};{\mathbf y})$ in a MIMO system, as presented in the following lemma. Let us first introduce some notations that make the overall understanding easier. Let ${\mathbf n}_e$ denote the equivalent to ${\mathbf n}$, real vector of length $2N_r$ derived from ${\mathbf n}$ by separating its real and imaginary parts as follows
\begin{equation}
{\mathbf n}_e = [n_{r1} ~n_{i1}\cdots n_{rN_r}~n_{iN_r}]^t,
\end{equation}
with $n_{rv},~n_{iv}$ being the values of the real, imaginary part of the $v$th ($1\leq v\leq N_r$) element of ${\mathbf n}$, respectively. Let us also define the real vector ${\mathbf v}(\{k_{rv},~k_{iv}\}_{v=1}^{N_r})$ of length $2N_r$ defined as follows
\begin{equation}
{\mathbf v}(\{k_{rv},~k_{iv}\}_{v=1}^{N_r}) = [v_{kr1}~v_{ki1},\cdots , v_{krN_r}~v_{kiN_r}]^t,
\end{equation}
with $k_{rv},~k_{iv}$ ($1\leq v \leq N_r$) being permutations of indexes in the set $\{1,2,\cdots,L\}$. Then the following lemma is true concerning the Gauss-Hermite approximation for $I({\mathbf x};{\mathbf y})$.
\begin{lemma}
For the MIMO channel model presented in (\ref{eq_1}), the Gauss-Hermite approximation for $I({\mathbf x};{\mathbf y})$ with $L$ nodes per receiving antenna is given as
\begin{equation}\small
\begin{split}
&I({\mathbf x};{\mathbf y})  \approx N_t\log_2(M) -\frac{N_r}{\log(2)} -\frac{1}{M^{N_t}}\sum_{k=1}^{M^{N_t}}{\hat f}_k,\\
\label{eq_PAPER}
\end{split}
\end{equation}
where
\begin{equation}
\begin{split}
{\hat f}_k =&  \left(\frac{1}{\pi}\right)^{N_r}\sum_{k_{r1}=1}^{L} \sum_{k_{i1}=1}^{L}\cdots \sum_{k_{rN_r}=1}^{L} \sum_{k_{iN_r}=1}^{L}  c(k_{r1})c(k_{i1})\cdots c(k_{rN_r})\\
&\times c(k_{iN_r})g_k(\sigma n_{{k_{r1}}}, \sigma n_{{k_{i1}}},\cdots,\sigma n_{{k_{rN_r}}}, \sigma n_{{k_{iN_r}}}), \label{eq_f}
\end{split}
\end{equation}
with
\begin{equation}
g_k(\sigma v_{{k_{r1}}}, \sigma v_{{k_{i1}}},\cdots,\sigma v_{{k_{rN_r}}}, \sigma v_{{k_{iN_r}}})
\end{equation}
being the value of the function
\begin{equation}
\log_2\left(\sum_{m}\exp (-\frac{1}{\sigma^2}||{\mathbf n}-{\mathbf H}{\mathbf G}({\mathbf x}_k-{\mathbf x}_m)||^2)\right) \label{eq_basic}
\end{equation}
evaluated at ${\mathbf n}_e =\sigma{\mathbf v}(\{k_{rv},~k_{iv}\}_{v=1}^{N_r})$.
\end{lemma}
The proof of this lemma is presented in Appendix A, together with a simplification available for this expression in the $N_t=N_r=2$ case.
Let us stress that, the presented novel application of the Gauss-Hermite quadrature in the MIMO model allows for efficient evaluation of $I({\mathbf x};{\mathbf y})$ for any channel matrix ${\mathbf H}$, and precoder ${\mathbf G}$, as required in the precoder optimization process, as explained below.

\section{GLOBAL OPTIMIZATION OVER ${\mathbf G}$ TOWARD MAXIMUM $I({\mathbf x};{\mathbf y})$ FOR QAM}
\subsection{Description of the Globally Optimal Precoder Method}
Similarly to \cite{Xiao}, we follow a block coordinate gradient ascent maximization method to find the solution to the optimization problem described in (\ref{eq_orig}), employing the virtual model of (\ref{eq_eq}). It is proven in \cite{Xiao} that $I({\mathbf x};{\mathbf y})$ is a concave function over
${\mathbf W}$ and ${\boldsymbol \Sigma}_G^2$. It thus becomes efficient to employ two different gradient ascent methods, one for ${\mathbf W}$, and another one for ${\boldsymbol \Sigma}_G^2$. We employ ${\boldsymbol \Theta}$ and ${\boldsymbol \Sigma} $ to denote ${\mathbf V}_G^h$ and ${\boldsymbol \Sigma}_G^2$, respectively, evaluated during the execution of the optimization algorithm.

The value of the step of each iteration over ${\mathbf W}$ and ${\boldsymbol \Sigma}_G^2$ is determined through backtracking line searches, one for each of the variables ${\mathbf W}$ and ${\boldsymbol \Sigma}_G^2$. We describe the backtracking line search for ${\mathbf W}$ first, then the one over ${\boldsymbol \Sigma}_G^2$. At each iteration of the globally optimal algorithm over ${\mathbf W}$, the gradient of $I({\mathbf x};{\mathbf y})$
over ${\mathbf W}$, $\nabla_{\mathbf W}I$, is required. We employ a novel method in determining $\nabla_{\mathbf W}I$ at each iteration, as described in detail in the next subsection.  We then select two parameters, $\alpha_1$ and $\beta_1$, both smaller than one and positive, and perform
the backtracking line search as follows: At each new trial, a parameter $t_1 > 0$ that represents the step size is updated by multiplying it with $\beta_1$. The initial value for $t_1$ is equal to 1. Then the algorithm checks if
\begin{equation}
I(\mathbf W + t\nabla_{\mathbf W}I) > I(\mathbf W) + \alpha_1 t ||\nabla_{\mathbf W}I||_F^2,
\end{equation}
where $||\nabla_{\mathbf W}I||_F$ is the Frobenius norm of $\nabla_{\mathbf W}I$ \cite{diff} and we used the notation $I({\mathbf W})$ to mean the value of $I({\mathbf x};{\mathbf y})$ at a fixed ${\mathbf W}$. If the condition is satisfied, the algorithm proceeds with calculating a new ${\boldsymbol \Theta}_{NEW}$ from ${\mathbf W}_{NEW} = {\mathbf W} + t_1 \nabla_{\mathbf W}I$, as follows:
\begin{equation}
{\mathbf W}_{NEW} = {\boldsymbol \Theta}_{NEW}^h {\boldsymbol \Sigma}_H^2 {\boldsymbol \Sigma}_G^2 {\boldsymbol \Theta}_{NEW} \label{eq_bck_1}
\end{equation}
by employing the eigenvalue decomposition (EVD) of ${\mathbf W}_{NEW}$, it updates ${\mathbf V}_G^h = {\boldsymbol \Theta}_{NEW}$,
and then it proceeds to the backtracking line search over ${\boldsymbol \Sigma}_G^2$. If the condition is not satisfied, the search updates $t_1$ to its new and smaller value, $\beta_1 t_1$, and
repeats the check on the condition, until the condition is satisfied or a maximum number of attempts in the first loop, $n_1$, has been reached. Then, the backtracking line search on ${\boldsymbol \Sigma}_G^2$ takes place in a fashion similar to the search described for ${\mathbf W}$, but with some ramifications. First, based on the second backtracking line search loop parameters
$\alpha_2$ and $\beta_2$, the backtracking line search is as follows:
At each new trial, a parameter $t_2 > 0$ that represents the step size is updated by multiplying it with $\beta_2$. The initial value for $t_2$ is equal to 1. Then the algorithm checks if
\begin{equation}
I(\mathbf W + t\nabla_{{\boldsymbol \Sigma}_G^2}I) > I(\mathbf W) + \alpha_2 t ||\nabla_{{\boldsymbol \Sigma}_G^2}I||_F^2, \label{eq_bck_2}
\end{equation}
where $||\nabla_{{\boldsymbol \Sigma}_G^2}I||_F$ is the Frobenius norm of $\nabla_{{\boldsymbol \Sigma}_G^2}I$ \cite{diff}. If the condition is satisfied, the algorithm proceeds with updating to the new ${\boldsymbol \Sigma}_{NEW}$, but after setting any negative terms in the main diagonal of ${\boldsymbol \Sigma}_{NEW}$ to zero and renormalizing the remaining main diagonal entries. If the condition is not satisfied, the search updates $t_2$ to its new and smaller value, $\beta_2 t_2$ and
repeats the check on the condition, until the condition is satisfied or a maximum number of attempts in the second loop, $n_2$ has been reached.
We thus see that there will in general be multiple evaluations of $I({\mathbf x};{\mathbf y})$, until the searches satisfy the conditions set or the maximum number of attempts allowed in a search has been reached. This explains the importance behind the requirement for an algorithm capable of
efficient calculation of $I({\mathbf x};{\mathbf y})$. In addition, as the parameters $\alpha_1,~\alpha_2$, $\beta_1,~\beta_2$ need to be optimized for faster and more efficient
execution of the globally optimal precoder optimization, this requirement becomes even more essential. Finally, the role of $n_1,n_2$ is also very important as when the number of attempts within each loop grows, the corresponding differential value of the parameter decreases and after a few attempts, the corresponding value of the step size is almost zero. By employing the proposed approach the
possibility of finding the globally optimal precoder for large QAM constellations with $M \leq 16$ or large MIMO configurations becomes reality, as our results demonstrate. The algorithm's pseudocode for a number of iterations $t$ is presented under the heading Algorithm 1.
\begin{algorithm}[h]
\caption{Global precoder optimization algorithm with $t$ iterations}\label{global}
\begin{algorithmic}[1]
\Procedure{Precoder}{${\boldsymbol \Sigma}_H$}
\While{$i \leq t $}
\State Determine ${\mathbf W}_{NEW}={\boldsymbol \Theta }^h {\boldsymbol \Sigma}^2  {\boldsymbol \Theta }$ through backtracking line search (\ref{eq_bck_1})
\State Set ${\mathbf V}_G^h = {\boldsymbol \Theta}$
\State Determine ${\mathbf W}_{NEW}={\mathbf V}_G {\boldsymbol \Sigma}_G^2  {\boldsymbol \Sigma}_H^2 {\mathbf V}_G^h$
\State Determine ${\boldsymbol \Sigma}_{NEW,G}^2$ through backtracking line search (\ref{eq_bck_2})
\State Set negative entries on the diagonal of ${\boldsymbol \Sigma}_{NEW,G}^2$ to zero
\State Normalize ${\boldsymbol \Sigma}_{NEW,G}^2$ to a trace equal to $N_t$
\State Set ${\boldsymbol \Sigma}_G = {\boldsymbol \Sigma}_{NEW,G}$
\State Determine ${\mathbf W}_{NEW}={\mathbf V}_G {\boldsymbol \Sigma}_G^2  {\boldsymbol \Sigma}_H^2 {\mathbf V}_G^h$
\State Set  ${\mathbf W} = {\mathbf W}_{NEW}$
\State Evaluate $I({\mathbf W})$
\EndWhile\label{globalendwhile}
\State \textbf{return} $I({\mathbf W})$
\EndProcedure
\end{algorithmic}
\end{algorithm}
\subsection{Determination of $\nabla_{\mathbf W}I,~\nabla_{{\boldsymbol \Sigma}_G^2}I$}
We first set ${\mathbf M} = {\mathbf W}^{\frac{1}{2}}$. Then, it is easy to see that $I$ is a function of ${\mathbf M}$ (see, e.g., \cite{Xiao3} where the notion of sufficient statistic is employed to show that $I({\mathbf x};{\mathbf y})$ depends on ${\mathbf W}$).
The derivation of $\nabla_{\mathbf W}I$ is presented in Appendix B. The proof is based on the following theorem\footnote{The theorem applies without loss of generality to the $N_t=N_r$ case. If $N_t\neq N_r$, then ${\boldsymbol \Sigma}_H$ needs to be either shrunk, or extended in size, by elimination or addition of zeros, respectively.}.
\setcounter{theorem}{0}
\begin{theorem} \label{thm_1}
Substituting ${\mathbf M}= {\mathbf V}_G {\boldsymbol \Sigma}_H {\boldsymbol \Sigma}_G {\mathbf V}_G^h= {\mathbf W}^{\frac{1}{2}}$ for ${\mathbf H}{\mathbf G}$ in (\ref{eq_PAPER}) results in the same value of $I({\mathbf x};{\mathbf y})$. In other words, since ${\mathbf M}$ is a function of ${\mathbf H},~{\mathbf G}$, ${\mathbf M}{\mathbf y}$ is a sufficient statistic for ${\mathbf y}$.
\end{theorem}
\begin{proof}
The proof of the theorem is simple. First, recall that the ``virtual'' channel model in (\ref{eq_eq}) is equivalent to the following model, which results by multiplying (\ref{eq_eq}) by the unitary matrix ${\mathbf V}_G$ on the left, resulting in
\begin{equation}
{\tilde {\mathbf y}}={\bf V}_G{\mathbf y} = {\bf V}_G{\boldsymbol \Sigma}_H {\boldsymbol \Sigma}_G {\bf V}_G^h {\mathbf x} + {\bf V}_G{\mathbf n}, \label{eq_eq_thm}
\end{equation}
where the modified noise term ${\bf V}_G{\mathbf n}$ has the same statistics with ${\mathbf n}$, because ${\bf V}_G$ is unitary. By applying the Gauss-Hermite approximation to (\ref{eq_eq_thm}), we see that we get the desired result, i.e., the value of $I({\mathbf x};{\mathbf y})$ remains the same, since both channel manifestations represent equivalent channels, i.e., the original one and its equivalent, thus their mutual information is the same. This completes the proof of the theorem.
\end{proof}
Note that using this theorem, an alternative proof of part of Theorem 1 in \cite{Xiao} can be developed, namely the fact that $I({\mathbf x};{\mathbf y})$ is only a function of ${\mathbf W}$, as ${\mathbf M}{\mathbf y}$ a sufficient statistic for ${\mathbf y}$ and ${\mathbf M}$ is a function of ${\mathbf W}$.

Assume without loss of generality that $N_t=N_r$.
The gradient of $I({\mathbf x};{\mathbf y})$ with respect to ${\mathbf M}$ can be found (see Appendix B for the derivation) from the Gauss-Hermite expression presented in (\ref{eq_PAPER}) as follows
\begin{equation}
\begin{split}
\nabla_{\mathbf M}I = &-\frac{1}{\log(2)} \frac{1}{M^{N_t}}\left(\frac{1}{\pi}\right)^{N_r}\sum_{k_{r1}=1}^{L} \sum_{k_{i1}=1}^{L}\cdots \sum_{k_{rN_r}=1}^{L} \sum_{k_{iN_r}=1}^{L} c(k_{r1})c(k_{i1})\cdots c(k_{rN_r})\\
 & \times c(k_{iN_r}){ {\mathbf R}}(\sigma v_{{k_{r1}}}, \sigma v_{{k_{i1}}},\cdots,\sigma v_{{k_{rN_r}}}, \sigma v_{{k_{iN_r}}}), \label{eq_grad_M}
\end{split}
\end{equation}
where ${ {\mathbf R}}(\sigma v_{{k_{r1}}}, \sigma v_{{k_{i1}}},\cdots,\sigma v_{{k_{rN_r}}}, \sigma v_{{k_{iN_r}}})$ is the value of the $N_t\times N_t$ matrix
\begin{equation}
\begin{split}
&\sum_{k}\frac{1}{\sum_{l}\exp (-\frac{1}{\sigma^2}||{\mathbf n}-{\mathbf M}({\mathbf x}_k-{\mathbf x}_l)||^2)}
\sum_{m}\left(\exp (-\frac{1}{\sigma^2}||{\mathbf n}-{\mathbf M}({\mathbf x}_k-{\mathbf x}_m)||^2)\right.\\
&\times\left. \left( ({\mathbf n}-{\mathbf M}({\mathbf x}_k-{\mathbf x}_m))({\mathbf x}_k-{\mathbf x}_m)^h + (({\mathbf n}-{\mathbf M}({\mathbf x}_k-{\mathbf x}_m))({\mathbf x}_k-{\mathbf x}_m)^h)^h \right) \vphantom{\exp (-\frac{1}{\sigma^2}||{\mathbf n}-{\mathbf M}({\mathbf x}_k-{\mathbf x}_m)||^2)}\right)
\end{split}
\end{equation}
evaluated at ${\mathbf n}_e =\sigma{\mathbf v}(\{k_{rv},~k_{iv}\}_{v=1}^{N_r})$.

The required $\nabla_{\mathbf W}I$ for the execution of the optimization process can be found from Appendix B as per the next lemma, using an easily proven equation.
Using the fact that for a Hermitian matrix such as ${\mathbf M}$, we need to add the Hermitian of the differential above in order to evaluate the actual gradient (see \cite{diff}), we get the desired result as follows (see Appendix B).
\begin{lemma}
For the MIMO channel model presented in (\ref{eq_1}), the Gauss-Hermite approximation allows to approximate $\nabla_{\mathbf W}I$ as follows.
\begin{equation}
\nabla_{\mathbf W}I \approx \mathrm{reshape}\left((\mathrm{\mathrm{vec}}(\nabla_{\mathbf M}I)^t \left(({\mathbf M}^*)\otimes{\mathbf I} + {\mathbf I}\otimes {\mathbf M}
 \right)^{-1}),N_t,N_t\right),
\end{equation}
where $\mathrm{reshape}({\mathbf A},k,n)$ is the standard reshape of a matrix ${\mathbf A}$ (with total number of elements $kn$) to a matrix with $k$ rows, $n$ columns, and
where $\otimes$ denotes Kronecker product of matrices. Standard \rm{reshape} emanates from the vector $\mathrm{vec}({\mathbf A})$ of matrix ${\mathbf A}$ which encompasses all columns of ${\mathbf A}$ starting from the leftmost one to the rightmost.
\end{lemma}
Then, as $I({\mathbf x};{\mathbf y})$ is a concave function of ${\mathbf W}$ \cite{Xiao3}, we can maximize over ${\mathbf W}$ in a straightforward way using closed form expressions. This is based on the fact that the approximated $I({\mathbf x};{\mathbf y})$ through the Gauss-Hermite approximation is very accurate, as shown in the next section.

The utilization of this lemma goes beyond the exploitation of the gradient of the mutual information in the precoder optimization algorithm. As it is well known from \cite{Xiao}$, \nabla_{\mathbf W}I = {\boldsymbol \Phi}$, where ${\boldsymbol \Phi} = {\mathbb E}\left\{ {\mathbb E}\{ ({\mathbf x}-{\mathbb E}({\mathbf x}|{\mathbf y}))({\mathbf x}-{\mathbb E}({\mathbf x}|{\mathbf y}))^h|{\mathbf y}\}\right\}$, is the minimum mean square error (MMSE) covariance matrix of the channel. Thus, by using the current lemma, an accurate estimate of the MMSE covariance matrix of the MIMO channel can be achieved. This is very useful especially when dealing with, e.g., SCSI cases where, the MMSE covariance matrix becomes instrumental in deriving the asymptotically optimal precoder \cite{SCSI}. Thus, based on the proposed Gauss-Hermite approximation, accurate, but otherwise simplified derivation of the MMSE covariance matrix of the channel becomes possible.

Finally, since from \cite{Xiao2} we have that $\nabla_{{\boldsymbol \Sigma}_G^2}I = \mathrm{diag}({\mathbf V}_G^h  \nabla_{\mathbf W}I {\mathbf V}_G{\boldsymbol \Sigma}_H^2),$ we can easily evaluate it through the procedure presented above.
\subsection{Complexity of the Globally Optimal Precoder Determination Method}
The complexity of the globally optimal precoder determination method depends mainly on the evaluations of $I({\mathbf x};{\mathbf y})$ as required by the optimization algorithm described above. By employing the Gauss-Hermite approximation described herein it becomes possible to significantly accelerate the evaluations of $I({\mathbf x};{\mathbf y})$ at each iteration of the optimization algorithm. However, the complexity is still high: Each evaluation of $I({\mathbf x};{\mathbf y})$ through the Gauss-Hermite approximation requires, as per the development in Appendix C, employing $L$ weights and $L$ nodes per each real and imaginary component of the noise vector ${\mathbf n}$, resulting in $2N_tL$ total weight and node dimensions. Then, evaluation of $f_k$ in (\ref{eq_f}) requires $2N_r$ nested ``DO'' loops, each of length $L$, resulting in $L^{2N_r}$ memory parameter values overall. For a good approximation in Gauss-Hermite quadrature, a value of $L \geq 3$ is required (please see the next section for relevant results) and thus the overall memory requirement becomes $3^{2N_r}$. As far as the computational complexity in the number of operations (including both (complex) summations and multiplications) involved in calculating $I({\mathbf x};{\mathbf y})$ through the Gauss-Hermite approximation, the corresponding complexity is shown in Appendix B to be $M^{N_t}(2L-1)^{(2N_r)} L^4 N_r (2N_t+N_r-1)$. Thus, for example, going from $M=16$ to $M=32$ QAM will result in about 4 times higher complexity with $N_t$ and $N_r$ held constant. On the other hand, increasing $N_r$ has an even more profound effect on the complexity, due to its more complicated presence in this complexity equation. For example, increasing $N_r$ from 2 to 3 while keeping all other parameters constant, will increase the complexity by a factor of $(2L-1)^2$ or for $L=3$ by 25 times. Our systematic numerical evaluations corroborate these numbers very closely. As we observe by comparing the figures in Table 1 to the numbers presented in \cite{Lozano}, we see that the Gauss-Hermite approach is about 7 times faster per iteration of the precoder for $M=16$ and PGP with groups of size $2\times2$, i.e., the presented approach is faster. This is one of the major improvements due to the proposed methodology.

\begin{table}\small
\caption{Computational Complexity Parameters of Globally Optimal Linear Precoder with QAM and Gauss-Hermit Quadrature}
\setcounter{table}{0}
\centering
\begin{tabular}{| c | c | |c| |c| |c|}
\hline \hline
$ $ &  &   & $I({\mathbf x};{\mathbf y})$  & ${\boldsymbol \nabla_{\mathbf W}I}$ \\
& ${\boldsymbol Complexity}$&${\boldsymbol Memory}$ & ${\boldsymbol Single~Run}$ & ${\boldsymbol Run}$\\
& &${\boldsymbol Requirement}$ & ${\boldsymbol CPU~sec}$ &${\boldsymbol CPU~sec}$\\
\hline
\bf{$M=16$} &$16^{N_t}(2L-1)^4 L^4 N_r (2N_t+N_r-1)$ &$L^{2N_r}$ &.25 & .54 \\
\hline
\bf{$M=32$} &$32^{N_t}(2L-1)^4 L^4 N_r (2N_t+N_r-1)$ &$L^{2N_r}$& 5.4&12.6   \\
\hline
\bf{$M=64$} &$64^{N_t}(2L-1)^4 L^4 N_r (2N_t+N_r-1)$ &$L^{2N_r}$& 50.20&110.60   \\
\hline
\end{tabular}
\end{table}

\section{{NUMERICAL RESULTS}}
The results presented in this subsection employ QAM with $16, ~32$, or $64$ constellation sizes. We employ MIMO systems with $N_t=N_r=2$ when global precoding optimization is performed. We have used an $L=3$ Gauss-Hermite approximation which results in $3^{2N_r}$ total nodes due to MIMO. The implementation of the globally optimizing methodology is performed by employing two
backtracking line searches, one for ${\mathbf W}$ and another one for ${\boldsymbol \Sigma}_G ^2$ at each iteration, in a fashion similar to \cite{Xiao3}. For the results presented, it is worth mentioning that only a few iterations (e.g., typically $<8$) are required to converge to the optimal solution results as presented in this paper. We apply the complexity reducing method of PGP \cite{TE_TWC} which offers semi-optimal results under exponentially lower transmitter and receiver complexity \cite{TE_TWC}. PGP divides the transmitting and receiving antennas into independent groups, thus achieving a much simpler detector structure while the precoder search is also dramatically reduced as well. Finally, we address both the CSIT and SCSI cases.

 We divide this section into five subsections. In the first subsection, we examine the accuracy of the proposed Gauss-Hermite approximation and provide a comparison with the lower bound technique presented in \cite{Xiao2}. In the second subsection we show results for the globally optimal precoder for $M=16,~32$ based on the approximation in conjunction with independent Gaussian channels and CSIT. In the third subsection we present results for SCSI channels similar to the ones in 3GPP SCM \cite{SCM}, with antenna size up to $100\times 100$ with PGP and modulation size $M=16,~32$. In the fourth subsection, we present results for CSIT jointly with PGP and higher size of antennas and modulation. Finally, in the last subsection, we present results for a MIMO system with $100$ base station antennas and $4$ user antennas with $M=16,~64$.
\subsection{Accuracy of the Gauss-Hermite approximation technique}
In Fig. 1 we present results for $I$ through simulation, the Gauss-Hermite approximation (GH) with $L=3$, and the lower bound developed in \cite{Xiao2} versus the signal-to-noise ratio per bit ($SNR_b$) in dB, for QAM with $M=16$, and for the commonly used channel \cite{Giann_max_div,Xiao}, $${\mathbf H}_1 = \left[ \begin{array} { c c}
2 & 1  \\
1 & 1
\end{array}
\right].$$
In the same figure, we also show a lower bound for $I({\mathbf x};{\mathbf y})$ which appeared in \cite{Xiao2}. We see excellent accuracy for the approximation, i.e., no observable difference between the Gauss-Hermite approximation and the simulations, over all $SNR_b$ values.
\begin{figure}[h]
\centering
\setcounter{figure}{0}
\includegraphics[height=2.4in,width=3.5in]{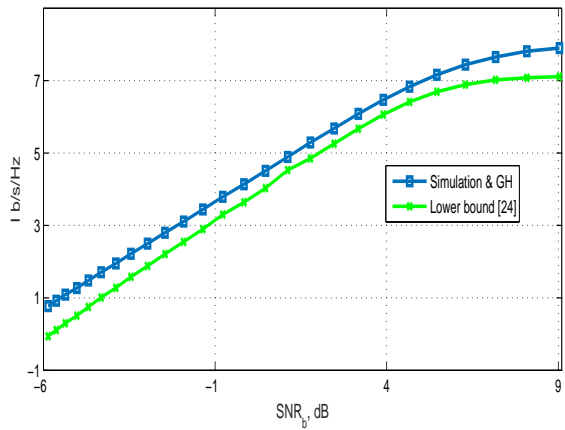}
	\caption{Results for $I({\mathbf x};{\mathbf y})$ without precoding for the ${\mathbf H}_1$ channel and QAM $M=16$ modulation.}
\end{figure}
In Fig. 2 we present the corresponding results for $M=32$ QAM in conjunction with the randomly generated channel
$${\mathbf H}_2 = \left[ \begin{array} { c c}
1.98+j0.12 & 0.0124-j0.0016  \\
-0.2487-j0.0314 & 0.0992-j0.1
\end{array}
\right],$$
with the same type of behavior as before, i.e., the Gauss-Hermite approximation offers excellent accuracy and that the lower bound is lagging behind in performance, albeit by less than $N_r(1/\log(2)-1)\approx 0.88,$ which is the shift introduced in \cite{Xiao2} in order to approximate $I$ very closely for QPSK modulation.
\begin{figure}[h]
\centering
\setcounter{figure}{1}
\includegraphics[height=2.4in,width=3.5in]{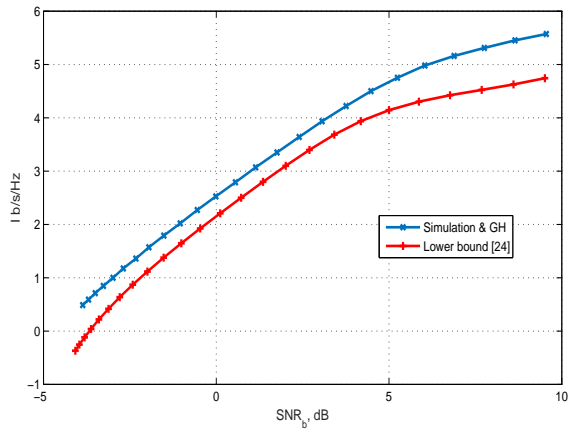}
	\caption{Results for $I({\mathbf x};{\mathbf y})$ without precoding for the ${\mathbf H}_2$ channel and QAM $M=32$ modulation.}
\end{figure}
\subsection{Results for Globally Optimal Precoding for Gaussian Channels and CSIT}
In Fig. 3 we show results for the globally optimal precoder based on the Gauss-Hermite approximation with $L=3$, and the Maximum Diversity Precoder (MDP) presented in  \cite{Giann_max_div}, for $M=16$, for ${\mathbf H}_1$. We observe that the optimal precoder offers significant utilization gains in the low SNR region, while its gain diminishes in the higher SNR region. For example, at $SNR_b=-4~dB,$ there is about $0.6~b/s/Hz$ gain attainable with the globally optimal precoder over its MDP and no-precoding counterparts, which represents a $30\%$ gain. Also, contrary to the BPSK/QPSK modulation case presented in \cite{Xiao,Xiao2}, the MDP precoder offers no significant gain over the no-precoding case, a very important difference. 
\begin{figure}[h]
\centering
\setcounter{figure}{2}
\includegraphics[height=2.4in,width=3.5in]{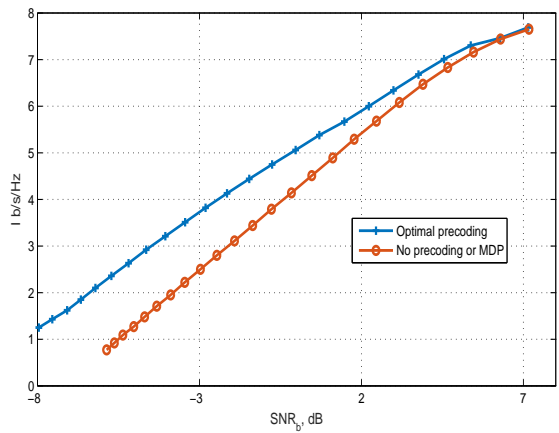}
	\caption{$I({\mathbf x};{\mathbf y})$ results for optimal precoding, MDP, and no-precoding cases for a $2\times 2$ MIMO system and QAM $M=16$ modulation.}
\end{figure}

In Fig. 4 we present results for $I({\mathbf x};{\mathbf y})$, for the same channel, MDP, and no precoding, $N_t=N_r=2$, and $M=32$.
There is no noticeable improvement offered by MDP over the no-precoding case, similarly to the $M=16$ QAM case presented above. Clearly, the same fundamental conclusions as in the $M=16$ case hold true. The offered gain of the globally optimal precoder in low SNR is still around $50\%$ over its no- precoding and MDP counterparts. 
\begin{figure}[h]
\centering
\setcounter{figure}{3}
\includegraphics[height=2.4in,width=3.5in]{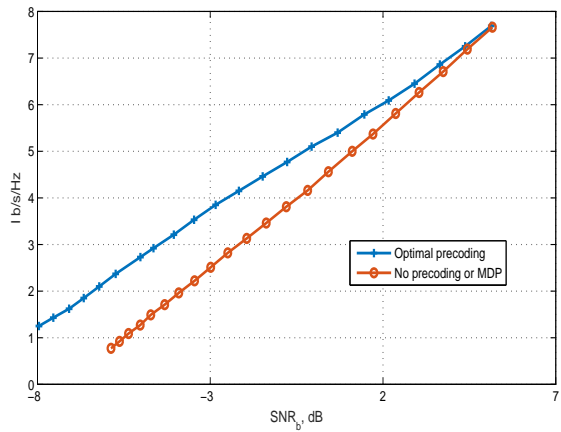}
	\caption{$I({\mathbf x};{\mathbf y})$ results for optimal precoding, MDP, and no-precoding cases for a $2\times 2$ ${\mathbf H}_1$ MIMO system and QAM $M=32$ modulation.}
\end{figure}
We next present same type of results for ${\mathbf H}_2$. In Fig. 5 we observe that for this type of channel, the gains achieved by the globally optimal precoder are significantly higher and that in the high $SNR$ regime the no-precoding case cannot achieve the maximum mutual information given by $N_t\log_2(M)=10$, i.e., it becomes saturated. We will see that for certain channels, which we call saturated channels, this type of behavior is also observed for the large MIMO channel configurations in both CSIT and SCSI channel cases.
\begin{figure}[h]
\centering
\setcounter{figure}{4}
\includegraphics[height=2.4in,width=3.5in]{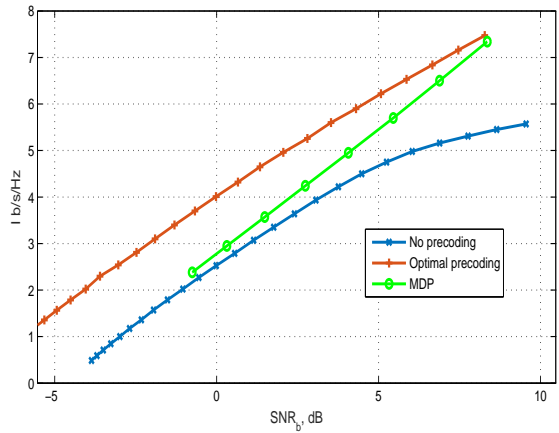}
	\caption{$I({\mathbf x};{\mathbf y})$ results for optimal precoding, MDP, and no-precoding cases for a $2\times 2$ ${\mathbf H}_2$ MIMO system and QAM $M=32$ modulation.}
\end{figure}

\subsection{Results for SCSI in Conjunction with PGP}
We first consider an $N_t=N_r=20$ SCM urban channel with $M=32$. We create the large size asymptotic approximation no-precoding results for a SCSI channel based on \cite{SCSI}. For precoding to be efficient, but otherwise realistic, we employ PGP \cite{TE_TWC} with 10 groups of size $2\times 2$ each. Due to the nature of this channel, we observe in Fig. 6 that the no-precoding case saturates at $I({\mathbf x};{\mathbf y})=70~b/s/Hz$. We see very significant gains offered by PGP in the high $SNR_b$ regime. The PGP system attains the full capacity available in high $SNR_b$.
\begin{figure}[h]
\centering
\setcounter{figure}{5}
\includegraphics[height=2.4in,width=3.5in]{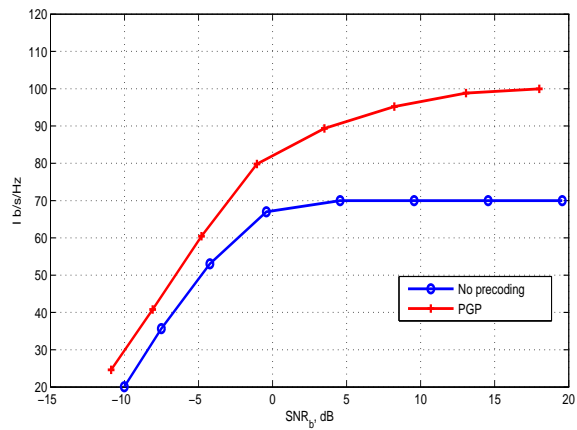}
	\caption{$I({\mathbf x};{\mathbf y})$ results for PGP and no-precoding cases for a $20\times 20$ ${\mathbf H}$ SCSI MIMO system and QAM $M=32$ modulation.}
\end{figure}
In Fig. 7 we present results for PGP versus a no-precoding SCM channel with $N_t=N_r=100$ and $M=16$. To the best of our knowledge, results for such large MIMO configurations are not available in the literature.
\begin{figure}[h]
\centering
\setcounter{figure}{6}
\includegraphics[height=2.4in,width=3.5in]{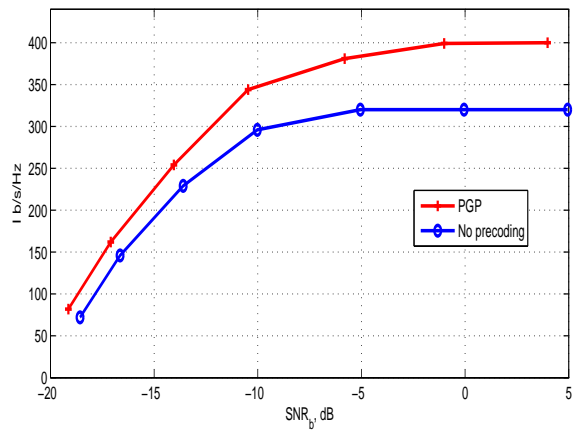}
	\caption{$I({\mathbf x};{\mathbf y})$ results for PGP and no-precoding cases for a $100\times 100$ ${\mathbf H}$ SCSI MIMO system and QAM $M=16$ modulation.}
\end{figure}
Similar to the previous results, we observe high information rate gains in the high $SNR_b$ regime as the PGP system achieves the full capacity of $400~b/s/Hz$ while the no-precoding scheme saturates at $320~b/s/Hz$. The PGP system employed uses $50$ groups of size $2\times 2$ each.
In Fig. 8 we present results for PGP versus a no-precoding SCM channel with $N_t=N_r=32$ and $M=16$, using the same SCM channel as in \cite{Lozano}. We observe that our proposed approach offers slightly better performance than the one employed in \cite{Lozano}, although it is faster, as it is explained in Table 1.
\begin{figure}[h]
\centering
\setcounter{figure}{7}
\includegraphics[height=2.4in,width=3.5in]{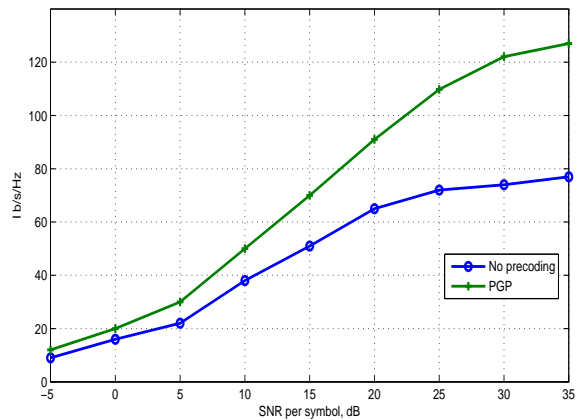}
	\caption{$I({\mathbf x};{\mathbf y})$ results for PGP and no-precoding cases for the $32\times 32$ ${\mathbf H}$ SCSI MIMO channel employed in \cite{Lozano} in conjunction with QAM $M=16$ modulation.}
\end{figure}
\subsection{Results for CSIT in Conjunction with  PGP}
For a $4\times 4$ channel
$${\mathbf H} = \left[ \begin{array} { c c c c}
 -1.5362 + 0.3151i &  0.5714 + 0.9123i  & 0.1394 - 0.3407i & -0.0085 + 0.0081i\\
  -1.5571 + 1.0171i & -0.3071 + 0.3765i & -0.3073 + 0.5680i & -0.0035 + 0.0041i\\
   0.4550 - 0.2484i &  0.7266 - 1.2195i &  0.0780 + 0.1645i & -0.0131 + 0.0008i\\
  -0.2278 + 3.1243i & -0.6890 - 0.3397i &  0.0175 - 0.2322i & -0.0045 - 0.0064i
   \end{array}
\right],$$
used with $M=64$ we get the no-precoding and the PGP results using 2 groups of size $2\times 2$ each depicted in Fig. 9. This example represents the corresponding CSIT case example that is similar to the SCM channels used in the previous subsection.
\begin{figure}[h]
\centering
\setcounter{figure}{8}
\includegraphics[height=2.4in,width=3.5in]{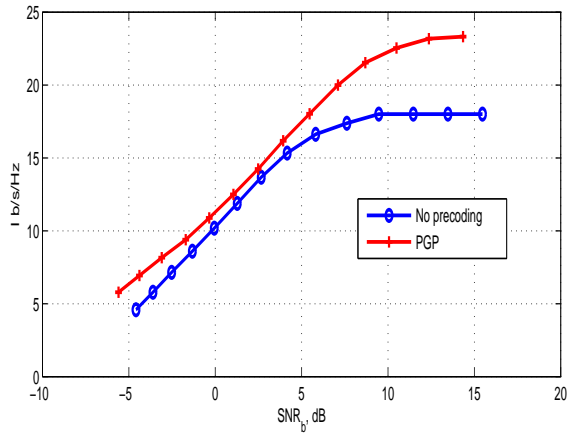}
	\caption{$I({\mathbf x};{\mathbf y})$ results for PGP and no-precoding cases for a $4\times 4$ ${\mathbf H}$ CSIT MIMO system and QAM $M=64$ modulation.}
\end{figure}
We observe very high gains of PGP over the no-precoding case in the high $SNR_b$ regime. To the best of our knowledge, this type of results for optimal precoding in conjunction with $M=64$ are not available in the literature.
Finally, in Fig. 10 we present results for an asymmetric randomly generated MIMO channel with $N_t=4,~N_r=10$, and $M=16$. PGP employs two groups of size $N_r=5,~N_t=2$ each. In the current scenario, we observe that significant gains are shown in the low $SNR_b$ regime, e.g., around $3~dB$ in $SNR_b$ lower than $-7~dB$.
\begin{figure}[h]
\centering
\setcounter{figure}{9}
\includegraphics[height=2.4in,width=3.5in]{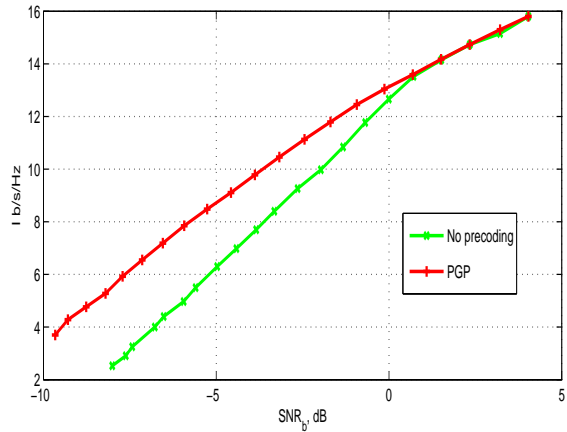}
	\caption{$I({\mathbf x};{\mathbf y})$ results for PGP and no-precoding cases for a randomly generated $10\times 4$ ${\mathbf H}$ CSIT MIMO system and QAM $M=16$ modulation.}
\end{figure}
\subsection{Results for Massive MIMO}
Massive MIMO \cite{Marzetta1, Marzetta2, Marzetta3} has attracted much interest recently, due to its potential to offer high data rates. We present results for the uplink, and downlink of a Massive MIMO system based on $100$ base station, $4$ user antennas, respectively, with $M=16,~64$, and for a Kronecker-based 3GPP SCM urban channel in a CSIT scenario in a single user configuration. Fig. 11 shows results for the $4\times100$ uplink of the system. We employ PGP to dramatically reduce the system complexity at the transmitter and receiver sites.
\begin{figure}[h]
\centering
\setcounter{figure}{10}
\includegraphics[height=2.4in,width=3.5in]{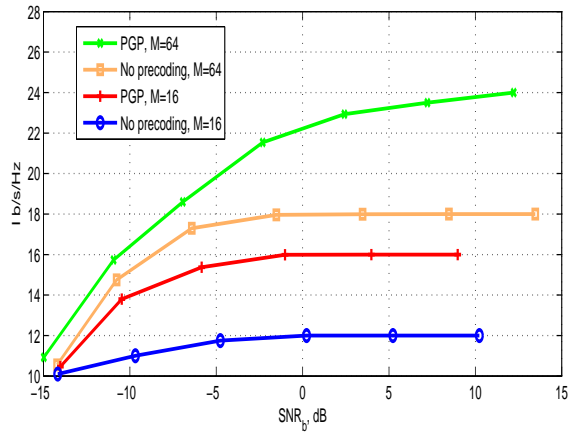}
	\caption{$I({\mathbf x};{\mathbf y})$ results for PGP and no-precoding cases for a randomly generated $100\times 4$ uplink ${\mathbf H}$ CSIT MIMO system and QAM $M=16$ modulation.}
\end{figure}
Under no precoding, the channel saturates and fails to meet the maximum possible mutual information of $16~b/s/Hz$, while with PGP the system clearly achieves the maximum mutual information rate, thus achieving high gains on the uplink in the high SNR regime. We stress the much higher throughput possible with $M=64$ over the $M=16$ case. For example, the no-precoding $M=16$ uplink significantly outperforms the PGP $M=16$ uplink. Second, the PGP $M=64$ uplink offers further gains by, e.g., achieving the maximum possible rate of $24~b/s/Hz$.
For the downlink, in Fig. 12 we show results where the no-precoding case operates under $100$ antenna inputs all correlated through the right eigenvectors of the channel, thus creating a very demanding environment at the user, due to the exponentially increasing maximum a posteriori (MAP) detector complexity \cite{TE_TWC}. On the other hand, employing PGP with only two input symbols per receiving antenna, i.e., with dramatically reduced decoding complexity, the PGP system achieves much higher throughput in the lower SNR regime, with SNR gain on the order of $10~dB$, albeit achieving a maximum of $32, ~48~b/s/Hz$ as there are a total of $8$ $M=16,~64$ QAM data symbols employed, respectively. We observe the superior performance of $M=64$ over its $M=16$ counterpart due to its increased constellation size. For example, at medium $SNR_b$, e.g., $SNR_b =4$, the $M=64$ PGP scheme achieves $45\%$ higher throughput that the $M=16$ one, a significant improvement. We would also like to emphasize that the no-precoding scheme requires a very high exponential MAP detector complexity, on the order of $M^{100}$, while for the low-SNR-superior PGP, this complexity is on the order of $M^2$ only. Thus, even in the higher SNR region where the no-precoding scheme can achieve a higher throughput, the complexity required at the user site becomes prohibitive. This demonstrates the superiority of PGP on the Massive MIMO downlink. On the other hand, in lower SNR, the PGP scheme achieves both much higher throughput with simultaneously exponentially lower MAP detector complexity at the user site detector. Note that the performance provided by PGP on the downlink depends
on the number of symbols processed jointly. This number is currently
limited by the computational complexity available. This limitation
does not occur on the uplink since in that case $N_t < N_r$.
\begin{figure}[h]
\centering
\setcounter{figure}{11}
\includegraphics[height=2.4in,width=3.5in]{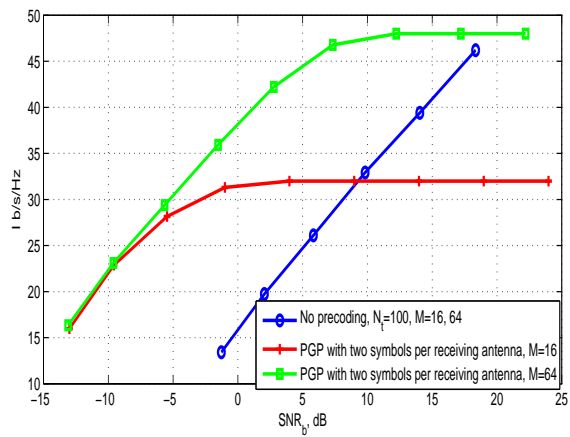}
	\caption{$I({\mathbf x};{\mathbf y})$ results for PGP and no-precoding cases for a randomly generated $4\times 100$ ${\mathbf H}$ downlink CSIT MIMO system and QAM $M=16$ modulation.}
\end{figure}
\section{{CONCLUSIONS}}
In this paper, the problem of designing a linear precoder for MIMO systems toward mutual information maximization is addressed for QAM with $M\geq 16$ and in conjunction with large MIMO system size. A major obstacle toward this goal is a lack of efficient techniques for evaluating
$I({\mathbf x},{\mathbf y})$ and its derivatives. We have presented a novel solution to this problem based on the Gauss-Hermite quadrature. We then applied a global optimization
framework to derive the globally optimal precoder for the case of QAM with $M=16$ and $32$ and small antenna size configurations. We showed that under CSIT in this case, significant gains are available for the lower SNR range over no precoding, or MDP. We showed that for the standard $2\times2$ channel, although the globally optimal precoder offers significant gains over MDP and the no-precoder configurations in the low SNR region, it fails to offer gains as high as $1~b/s/Hz$. However, we demonstrated that by employing
another $2\times2$ channel gains as high as $1.4~b/s/Hz$ are possible. For systems of large MIMO configurations, we applied the complexity-simplifying PGP concept \cite{TE_TWC} to derive semi-optimal precoding results. Under SCSI, we showed that for urban 3GPP SCM channels, an interesting saturation effect in the no-precoding case takes place, while the semi-optimal PGP precoder offers dramatically better results in this case while it does not experience any saturation as the SNR increases, e.g., it achieves the maximum information rate,
$I({\mathbf x};{\mathbf y}) = N_t\log_2(M)$ at high SNR. Furthermore, we applied the same Gauss-Hermite approximation approach to CSIT with a large number of antenna with the same success. We showed that for specific type of channels similar to urban 3GPP SCM \cite{SCM}, the PGP approach offers very high gains over the no-precoding case in the high $SNR_b$ regime. Finally, we considered a Massive MIMO scenario in conjunction with CSIT and showed that by carefully designing the downlink and uplink precoders, the methodology shows very high gains, especially on the downlink, although it employs an exponentially simpler MAP detector at the user site.

Based on the evidence presented, the novel application of the Gauss-Hermite quadrature rule in the MIMO scenario allows for generalizing the interesting results presented in \cite{Xiao,SCSI} to the QAM case with ease. Because of the simplification achieved by the combination of PGP and the Gauss-Hermite approximation, we were able to derive results with, e.g., $N_t=N_r=100$ as well as with $M=64$ efficiently. In addition, the presented Gauss-Hermite approximation offers important simplification in the evaluation of the MMSE covariance matrix of the MIMO channel which is required in, among other areas, the SCSI equivalent channel determination \cite{SCSI}.

\appendices
\section{Gauss-Hermite Quadrature Approximation in MIMO Input Output Mutual Information}
$I({\mathbf x};{\mathbf y}) = H({\mathbf x})-H({\mathbf x}|{\mathbf y})= N_t\log_2(M)-H({\mathbf x}|{\mathbf y})$, where the conditional entropy, $H({\mathbf x}|{\mathbf y})$ can be written as \cite{Xiao}
\begin{equation}\small
\begin{split}
H({\mathbf x}|{\mathbf y})  &= \frac{N_r}{\log(2)} +\frac{1}{M^{N_t}}\sum_k{\mathbb E}_{\mathbf n}\left( \log_2\left(\sum_{m}\exp (-\frac{1}{\sigma^2}||{\mathbf n}-{\mathbf H}{\mathbf G}({\mathbf x}_k-{\mathbf x}_m)||^2)\right) \right)\\
&=\frac{N_r}{\log(2)} +\frac{1}{M^{N_t}}\sum_{k}\int_{-\infty}^{+\infty} {\cal N}_c({\mathbf n}|{\mathbf 0},\sigma^2 {\mathbf I})
\log_2\left(\sum_{m}\exp (-\frac{1}{\sigma^2}||{\mathbf n}-{\mathbf H}{\mathbf G}({\mathbf x}_k-{\mathbf x}_m)||^2)\right) d{\mathbf n},
\label{eq_first}
\end{split}
\end{equation}
where ${\cal N}_c({\mathbf n}|{\mathbf 0},\sigma^2 {\mathbf I})$ represents the probability density function (pdf) of the circularly symmetric complex random vector due to AWGN. Let us define
\begin{equation}
\begin{split}
f_k \doteq \int_{-\infty}^{+\infty} {\cal N}_c({\mathbf n}|{\mathbf 0},\sigma^2 {\mathbf I})
\log_2\left(\sum_{{\mathbf x}_m}\exp (-\frac{1}{\sigma^2}||{\mathbf n}-{\mathbf H}{\mathbf G}({\mathbf x}_k-{\mathbf x}_m)||^2)\right) d{\mathbf n} \label{eq_special}.
\end{split}
\end{equation}
Since ${\mathbf n}$ has independent components over the different receiving antennas, and over the real and imaginary dimensions, the integral above can be partitioned into $2N_r$ real integrals in tandem, in the following manner: Define by $n_{rv}, n_{iv}$, with $v=1,\cdots,N_r$, the $v$th receiving antenna real and imaginary noise component, respectively. Also define by $({\mathbf H}{\mathbf G}({\mathbf x}_k-{\mathbf x}_m))_{rv}$ and $({\mathbf H}{\mathbf G}({\mathbf x}_k-{\mathbf x}_m))_{iv}$, the $v$th receiving antenna real and imaginary component of $({\mathbf H}{\mathbf G}({\mathbf x}_k-{\mathbf x}_m))$, respectively. We then have
\begin{equation}
{\cal N}_c({\mathbf n}|{\mathbf 0},\sigma^2 {\mathbf I}) = \frac{1}{\pi^{N_r}\sigma^{2N_r}}\exp(-\frac{\sum_l n_{rv}^2 + n_{iv}^2}{\sigma^2}),
\end{equation}
\begin{equation}
d{\mathbf n} = \prod_{v=1}^{N_r}dn_{rv}dn_{iv},
\end{equation}
and
\begin{equation}
\begin{split}
&\sum_{m}\exp (-\frac{1}{\sigma^2}||{\mathbf n}-{\mathbf H}{\mathbf G}({\mathbf x}_k-{\mathbf x}_m)||^2) \\
&=\sum_{m}\exp (-\frac{1}{\sigma^2}(\sum_v({ n}_{rv}-({\mathbf H}{\mathbf G}({\mathbf x}_k-{\mathbf x}_m))_{rv})^2\\
&+\sum_v({ n}_{iv}-({\mathbf H}{\mathbf G}({\mathbf x}_k-{\mathbf x}_m))_{iv})^2)).
\end{split}
\end{equation}
The Gauss-Hermite quadrature is as follows:
\begin{equation}
\int_{-\infty}^{+\infty}\exp(-x^2)f(x)dx \approx \sum_{l=1}^L c(l) f(v_l),
\end{equation}
for any real function $f(x)$, and with vector ${\mathbf c} = [c(1)\cdots c(L)]^T$ being the ``weights,'' and $v_l$ are the ``nodes'' of the approximation. The approximation is based on the following weights and nodes \cite{G_H}
\begin{equation}
c(l) = \frac{2^{L-1}L! {\sqrt2\pi}}{L^2 (H_{L-1}(v_l))^2}
\end{equation}
where $H_{L-1}(x) = (-1)^{L-1} \exp(x^2) \frac{d^{L-1}}{dx^{L-1}}(\exp(-x^2))$ is the $(L-1)$-th order Hermitian polynomial, and the value of the node
$v_l$ equals the root of $H_L(x)$ for $l=1,2,\cdots,L$.

Applying the Gauss-Hermite quadrature $2N_r$ times in tandem to the integral in (\ref{eq_special}), and after changing variables, we get that
\begin{equation}
\begin{split}
f_k \approx {\hat f}_k =&  \left(\frac{1}{\pi}\right)^{N_r}\sum_{k_{r1}=1}^{L} \sum_{k_{i1}=1}^{L}\cdots \sum_{k_{rN_r}=1}^{L} \sum_{k_{iN_r}=1}^{L} c(k_{r1})c(k_{i1})\cdots c(k_{rN_r})\\
& \times c(k_{iN_r})g_k(\sigma n_{{k_{r1}}}, \sigma n_{{k_{i1}}},\cdots,\sigma n_{{k_{rN_r}}}, \sigma n_{{k_{iN_r}}}), \label{eq_f}
\end{split}
\end{equation}
where
\begin{equation}
g_k(\sigma n_{{k_{r1}}}, \sigma n_{{k_{i1}}},\cdots,\sigma n_{{k_{rN_r}}}, \sigma n_{{k_{iN_r}}})
\end{equation}
is the value of the function (from (\ref{eq_f}))
\begin{equation}
\log_2\left(\sum_{m}\exp (-\frac{1}{\sigma^2}||{\mathbf n}-{\mathbf H}{\mathbf G}({\mathbf x}_k-{\mathbf x}_m)||^2)\right) \label{eq_basic}
\end{equation}
evaluated at ${\mathbf n}_e =\sigma{\mathbf v}(\{k_{rv},~k_{iv}\}_{v=1}^{N_r})$.

For the special case $N_r=2$, (\ref{eq_f}) becomes
\begin{equation}
\begin{split}
{\hat f}_k =&  \sum_{k_{r1}=1}^{L} \sum_{k_{i1}=1}^{L} \sum_{k_{r2}=1}^{L} \sum_{k_{i2}=1}^{L} c(k_{r1})c(k_{i1}) \\
& \times c(k_{r2})c(k_{i2})g_k(\sigma n_{{k_{r1}}}, \sigma n_{{k_{i1}}},\sigma n_{{k_{r2}}}, \sigma n_{{k_{i2}}}),
\end{split}
\end{equation}
which using basic properties of bilinear forms can be rewritten as
\begin{equation}
\begin{split}
{\hat f}_k =&  \left(\frac{1}{\pi}\right)^2\sum_{k_{r1}=1}^{L} \sum_{k_{i1}=1}^{L} \sum_{k_{r2}=1}^{L} \sum_{k_{i2}=1}^{L} c(k_{r1})c(k_{i1}) c(k_{r2})\\
& \times c(k_{i2}) g_k(\sigma n{{k_{r1}}}, \sigma n_{{k_{i1}}},\sigma n_{{k_{r2}}}, \sigma n_{{k_{i2}}})\\
& = \left(\frac{1}{\pi}\right)^2{\mathbf c}^t {\mathbf F}_k {\mathbf c}, \label{eval}
\end{split}
\end{equation}
where
${\mathbf F}_k$ is an $L\times L$ matrix with $k_{r1},k_{i1}$ element equal to
\begin{equation}
{\mathbf F}_k[k_{r1},k_{i1}] = {\mathbf c}^t {\mathbf V}_k {\mathbf c}, \label{eq_F}
\end{equation}
with ${\mathbf V}_k$ being an $L \times L$ matrix with $k_{r2},k_{i2}$ element equal to
\begin{equation}
{\mathbf V}_k[k_{r2},k_{i2}] = g_k(\sigma n_{{k_{r1}}}, \sigma n_{{k_{i1}}},\sigma n_{{k_{r2}}}, \sigma n_{{k_{i2}}}), \label{eq_compl}
\end{equation}
so that the different $f_k$ can be approximated efficiently, and then summing them over the different ${\mathbf x}_k$, as per (\ref{eq_first}), we get an
approximation for $H({\mathbf x}|{\mathbf y})$,
\begin{equation}\small
\begin{split}
H({\mathbf x}|{\mathbf y})  \approx & \frac{N_r}{\log(2)} +\frac{1}{\pi^2M^{N_t}}\sum_{k}{\hat f}_k\\
=&\frac{N_r}{\log(2)} +\frac{1}{\pi^2 M^{N_t}}\sum_{k}{\mathbf c}^t {\mathbf V}_k {\mathbf c}
=\frac{N_r}{\log(2)} +\frac{1}{\pi^2M^{N_t}}{\mathbf c}^t (\sum_{k}{\mathbf V}_k) {\mathbf c}
=\frac{N_r}{\log(2)} +\frac{1}{\pi^2M^{N_t}}{\mathbf c}^t {\mathbf V} {\mathbf c},
\end{split}
\end{equation}
where ${\mathbf V}  \doteq \sum_{k}{\mathbf V}_k.$


\section{Derivation of $\nabla_{\mathbf W}I$ through the Gauss-Hermite approximation}
Without loss of generality, let's assume that $N_t=N_r$. Using Theorem \ref{thm_1}, we can write by using the Gauss-Hermite approximation with ${\mathbf M}$ instead of ${\mathbf H}{\mathbf G}$,
\begin{equation}
I({\mathbf x};{\mathbf y})  \approx N_t\log_2(M) -\frac{N_r}{\log(2)} -\frac{1}{M^{N_t}}\sum_{k}{\hat f}_k \label{eq_PAPER_M}.
\end{equation}

In order to derive the gradient of $I({\mathbf x};{\mathbf y})$ with respect to ${\mathbf W}$, we first derive the gradient of $I({\mathbf x};{\mathbf y})$ with respect to ${\mathbf M}^*$. Start with the differential of $I({\mathbf x};{\mathbf y})$ with respect to ${\mathbf M}^*$ in (\ref{eq_PAPER_M}) and approximate the $f_k$ by ${\hat f}_k$,
\begin{equation}
\begin{split}
d_{{\mathbf M}^{*}}I({\mathbf x};{\mathbf y})  \approx &   -\frac{1}{M^{N_t}}d_{{\mathbf M}^{*}}\left(\sum_{k}{\hat f}_k\right)\\
=&  -\frac{1}{M^{N_t}}\left(\frac{1}{\pi}\right)^{N_r}\sum_{k_{r1}=1}^{L} \sum_{k_{i1}=1}^{L}\cdots \sum_{k_{rN_r}=1}^{L} \sum_{k_{iN_r}=1}^{L} c(k_{r1})c(k_{i1})\cdots c(k_{rN_r})\\
&\times c(k_{iN_r})\sum_{k}d_{{\mathbf M}^{*}}\left(g_k(\sigma n_{{k_{r1}}}, \sigma n_{{k_{i1}}},\cdots,\sigma n_{{k_{rN_r}}}, \sigma n_{{k_{iN_r}}})\right).
\label{eq_PAPER_M1}
\end{split}
\end{equation}
Taking into account that $g_k(\sigma n_{{k_{r1}}}, \sigma n_{{k_{i1}}},\cdots,\sigma n_{{k_{rN_r}}}, \sigma n_{{k_{iN_r}}})$ is the value of
$$\log_2\left(\sum_{m}\exp (-\frac{1}{\sigma^2}||{\mathbf n}-{\mathbf M}({\mathbf x}_k-{\mathbf x}_m)||^2)\right)$$
evaluated at ${\mathbf n}_e =\sigma{\mathbf v}(\{k_{rv},~k_{iv}\}_{v=1}^{N_r})$, we can develop (\ref{eq_PAPER_M1}) further, by using well-known results \cite{diff}, as follows
\begin{equation}
\begin{split}
d_{{\mathbf M}^{*}}I({\mathbf x};{\mathbf y})  \approx & -\frac{1}{\log(2)} \frac{1}{M^{N_t}}\left(\frac{1}{\pi}\right)^{N_r}\sum_{k_{r1}=1}^{L} \sum_{k_{i1}=1}^{L}\cdots \sum_{k_{rN_r}=1}^{L} \sum_{k_{iN_r}=1}^{L} c(k_{r1})c(k_{i1})\cdots c(k_{rN_r})\\
  & \times c(k_{iN_r})\mathrm{tr}\left(\{{\mathbf R}(\sigma n_{r1,{k_{r1}}}, \sigma n_{i1,{k_{i1}}},\cdots,\sigma n_{rN_r,{k_{rN_r}}}, \sigma n_{iN_r,{k_{iN_r}}})\}^t d{\mathbf M^*}\right),\label{eq_PAPER_M2}
\end{split}
\end{equation}
where ${\mathbf R}(\sigma n_{{k_{r1}}}, \sigma n_{{k_{i1}}},\cdots,\sigma n_{{k_{rN_r}}}, \sigma n_{{k_{iN_r}}})$ is the value of
\begin{equation}
\begin{split}
&\sum_{k}\frac{1}{\sum_{l}\exp (-\frac{1}{\sigma^2}||{\mathbf n}-{\mathbf M}({\mathbf x}_k-{\mathbf x}_l)||^2)}
\sum_{m}\left( \exp (-\frac{1}{\sigma^2}||{\mathbf n}-{\mathbf M}({\mathbf x}_k-{\mathbf x}_m)||^2)\right.\\
&\left. \times \left( ({\mathbf n}-{\mathbf M}({\mathbf x}_k-{\mathbf x}_m))({\mathbf x}_k-{\mathbf x}_m)^h  \right)\vphantom{\exp (-\frac{1}{\sigma^2}||{\mathbf n}-{\mathbf M}({\mathbf x}_k-{\mathbf x}_m)||^2)}\right)
\end{split}
\end{equation}
evaluated at ${\mathbf n}_e =\sigma{\mathbf v}(\{k_{rv},~k_{iv}\}_{v=1}^{N_r})$. In this derivation we have used the fact that
\begin{equation}
\begin{split}
&d_{{\mathbf M}}\left(\exp(-\frac{||{\mathbf n}-{\mathbf M}({\mathbf x}_k-{\mathbf x}_m)||^2}{\sigma^2})\right)=\frac{1}{\sigma^2}\times \\
&\exp(-\frac{||n-{\mathbf M}({\mathbf x}_k-{\mathbf x}_m)||^2}{\sigma^2})({\mathbf n}-{\mathbf M}({\mathbf x}_k-{\mathbf x}_m))({\mathbf x}_k-{\mathbf x}_m)^h.
\end{split}
\end{equation}
Using the fact that for a Hermitian matrix such as ${\mathbf M}$, we need to add the Hermitian of the differential above in order to evaluate the actual gradient (see \cite{diff}), we get (\ref{eq_grad_M}).

Now, since ${\mathbf W} = {\mathbf M}^2$, by taking vectors of the matrices on each side of this equation and applying some identities from \cite{diff}, we get that
\begin{equation}
d\mathrm{vec}({\mathbf W}) = \left(({\mathbf M}^*)\otimes{\mathbf I} + {\mathbf I}\otimes {\mathbf M}
 \right)d\mathrm{vec}({\mathbf M}), \label{eq_final}
\end{equation}
where the notation $\mathrm{vec}({\mathbf A})$ denotes the vector found from matrix ${\mathbf A}$ by taking its columns one at a time, starting from the leftmost one. Since ${\mathbf M} = {\mathbf W}^{\frac{1}{2}}$, it is straightforward to derive (\ref{eq_special}) by employing (\ref{eq_final}) as follows. The relationship between the gradient and the derivative is through reshaping the derivative row vector \cite{diff}, we can thus write
\begin{equation}
\nabla_{\mathbf W}I = \mathrm{reshape}\left((\mathrm{vec}(\nabla_{\mathbf M}I)^t \left(({\mathbf M}^*)\otimes{\mathbf I} + {\mathbf I}\otimes {\mathbf M}
 \right)^{-1}),N_t,N_t\right),
\end{equation}
where $\mathrm{reshape}({\mathbf A},k,n)$ is the standard reshape of matrix ${\mathbf A}$ with total elements $kn$ to a matrix with $k$ rows, $n$ columns.
\section{Evaluation of Computational Complexity in Calculating $I({\mathbf x};{\mathbf y})$ through Gauss-Hermite Quadrature}
Let's start with $N_r=2$. Then, the complexity involved in the Gauss-Hermite quadrature approximation of $I({\mathbf x};{\mathbf y})$ is determined by the one required in
the calculation of $f_k$ in (\ref{eval}). From (\ref{eq_F}), this is equal to $(2L-1)^2$ times the number of operations, including summations and products, required in evaluating each element of the matrix ${\mathbf V}_k$. Since ${\mathbf V}_k$ is a size $L\times L$ matrix with elements given in (\ref{eq_compl}), it can be seen from (\ref{eq_basic}) that the complexity of each element ${\mathbf V}_k[k_{r2},k_{i2}]$ is $L^2 N_r (2N_t+N_r-1)$. Since there are $M^{N_t}$ summation terms over $k$ (the size of the overall multiple input constellation), the total complexity becomes $(2L-1)^4 L^4 N_r (2N_t+N_r-1)$. For a general value of $N_r$, in a similar fashion, the corresponding complexity becomes $M^{N_t}(2L-1)^{(2N_r)} L^4 N_r (2N_t+N_r-1)$.
\bibliographystyle{IEEEtran}


\end{document}